\documentclass[11pt,fleqn,a4paper]{article}

\usepackage{amsmath,amssymb,amsthm}
\usepackage[mathscr]{eucal}

\allowdisplaybreaks

\newcommand{\p}{\partial}

\newcommand{\ve}{\varepsilon}
\newcommand{\FF}{\mathcal{F}}
\newcommand{\GG}{\mathcal{G}}
\newcommand{\PP}{\mathcal{P}}
\newcommand{\DDD}{\mathcal{D}}

\newtheorem{theorem}{Theorem}

\newtheorem{corollary}{Corollary}
\newtheorem{proposition}{Proposition}
{\theoremstyle{definition}
\newtheorem{definition}{Definition}

\newtheorem{remark}{Remark}
\newtheorem*{remark*}{Remark}
}

\newcommand{\todo}[1][\null]{\ensuremath{\clubsuit}}

\newcommand{\noprint}[1]{}
\newcommand{\checked}[1][\null]{\ensuremath{\boldsymbol{\surd}}}

\begin{document}

\par\noindent {\LARGE\bf
Algebraic method for finding equivalence groups
\par}

{\vspace{4mm}\par\noindent {\bf Alexander Bihlo$^\dag\, ^\ddag$, Elsa Dos Santos Cardoso-Bihlo$^\S$ and Roman O. Popovych$^\S\, ^\P$
} \par\vspace{2mm}\par}

{\vspace{2mm}\par\noindent {\it
$^{\dag}$~Department of Mathematics and Statistics, Memorial University of Newfoundland,\\ St.\ John's (NL) A1C 5S7, Canada
}}
{\vspace{2mm}\par\noindent {\it
$^{\ddag}$~Department of Mathematics, University of British Columbia,\\ Vancouver (BC) V6T 1Z2, Canada
}}
{\vspace{2mm}\par\noindent {\it
$^{\S}$~Wolfgang Pauli Institut, Universit\"at Wien, Oskar-Morgenstern-Platz 1, 1090 Wien, Austria
}}
{\vspace{2mm}\par\noindent {\it
$^\P$~Institute of Mathematics of NAS of Ukraine, 3 Tereshchenkivska Str., 01601 Kyiv, Ukraine
}}

{\vspace{2mm}\par\noindent {\it
\textup{E-mail:} abihlo@mun.ca, elsa.cardoso@univie.ac.at, rop@imath.kiev.ua
}\par}

{\vspace{5mm}\par\noindent\hspace*{5mm}\parbox{150mm}{\small
The algebraic method for computing the complete point symmetry group of a system of differential equations
is extended to finding the complete equivalence group of a class of such systems.
The extended method uses the knowledge of the corresponding equivalence algebra.
Two versions of the method are presented, where the first involves the automorphism group of this algebra and the second is based on a list of its megaideals.
We illustrate the megaideal-based version of the method with the computation
of the complete equivalence group of a class of nonlinear wave equations with applications in nonlinear elasticity.
}\par\vspace{4mm}}

\section{Introduction}

Classes of (systems of) differential equations are (systems of) differential equations that contain parameters, which may be constants or functions, and are collectively referred to as arbitrary elements. Most systems of differential equations of science and engineering involve such arbitrary elements and hence are classes rather than single systems. The wide occurrence of classes of differential equations is the reason why their study has a prominent place in group analysis of differential equations \cite{bihl11Dy,ovsi82Ay,popo06Ay,popo10Ay}. The classification of Lie symmetry properties of systems from a class depending on the values of the arbitrary elements is called the group classification problem for this class.

Equivalence transformations of a class of differential equations are point (resp.\ contact) transformations in the underlying space of independent and dependent variables, involved derivatives and arbitrary elements that map each system from the class to another system from the same class \cite{ovsi82Ay,popo10Ay}.
Such transformations play a central role in group classification of differential equations since
they can be used to select the simplest representative among similar systems.
Continuous equivalence transformations form a connected Lie group, which can thus be found using the infinitesimal method by computing the corresponding Lie algebra, the so-called equivalence algebra, i.e., the set of vector fields which generate one-parameter subgroups of equivalence transformations. Due to its algorithmic nature, this is a standard way for computing equivalence transformations in papers on group classification of differential equations~\cite{akha91Ay,ovsi82Ay}.

The main shortcoming of the infinitesimal method is that it misses discrete equivalence transformations. That is, not the complete equivalence group is obtained but only its continuous component. The situation is completely analogous to the construction of the maximal Lie invariance group for a single system of differential equations, which is carried out within the framework of the infinitesimal method; but in order to present the corresponding complete group of point symmetries, also discrete symmetry transformations have to be computed.

Finding complete point symmetry groups (and, analogously, complete equivalence groups) is a nontrivial task if the \emph{direct method} is applied \cite{bihl11Cy,king91By,king91Ay,king98Ay,popo10Ay}.
In the direct method one computes point symmetries without the detour of infinitesimal generators. The direct method hence also allows one to find discrete point symmetries and thus to construct the complete point symmetry group.%
\footnote{%
Discrete symmetries like alternating signs of independent and/or dependent variables can easily be found ad hoc, by empiric methods.
Even if such trivial symmetries exhaust a complete set of discrete symmetries that are independent up to composing with continuous symmetries,
which is often the case, the problem is to prove that this is really the case.
Such a proof cannot be realized within an empirical approach.
}
This method does not use the linearization of the determining equations, which enables the algorithmic computation of Lie symmetries, and thus usually one has to solve a system of nonlinear partial differential equations; which is a challenging endeavor.
For systems with a~nontrivial continuous component of the complete point symmetry group, a refinement for the pure direct method can be used, the so-called \emph{algebraic method}.
The algebraic method is simpler in that it allows one to restrict the principal form of point transformations before invoking the direct method. In particular, the solution of nonlinear systems of partial differential equations can often be completely avoided.

The central idea of the algebraic method for the computation of discrete symmetries of differential equations is the following:
Given a system of differential equations~$\mathcal L$, let~$\mathfrak g_{\mathcal L}$ denote its maximal Lie invariance algebra, which is known.
Push-forwards of vector fields defined on the corresponding space of independent and dependent variables
by point symmetries of~$\mathcal L$ induce automorphisms of~$\mathfrak g_{\mathcal L}$.
This property can be used in order to restrict the form of those point transformations that can be symmetries of~$\mathcal L$.
The restricted form can then be substituted into the system~$\mathcal L$ from which point one still proceeds with the direct method.
In fact, the algebraic method just simplifies the application of the direct method.

The original version of the algebraic method was proposed by Hydon~\cite{hydo98Ay,hydo00By,hydo00Ay}; see also~\cite{gray13Ay}.
It is based on the immediate use of the above property of point symmetries of~$\mathcal L$
and requires the explicit computation of the entire automorphism group of~$\mathfrak g_{\mathcal L}$;
this is why we call it the automorphism-based version of the algebraic method or, shortly, the \emph{automorphism-based method}.
The computation of the entire automorphism group is a comprehensive or even impossible task for infinite-dimensional Lie algebras.
Therefore, for a proper application of the automorphism-based method the algebra~$\mathfrak g_{\mathcal L}$ has to be finite dimensional,
which is a severe limitation of the applicability of this version of the algebraic method.

There are countless examples of systems of differential equations with infinite-dimensional maximal Lie invariance algebras, especially in hydrodynamics. This is why in~\cite{bihl11Cy} an essential refinement of the algebraic method was proposed. It is built around the notion of fully characteristic ideals (shortly called \emph{megaideals}); these are vector subspaces of a Lie algebra~$\mathfrak g$ that are invariant under any transformation from the automorphism group of $\mathfrak g$ \cite{bihl11Cy,card12Ay,popo03Ay}, see also \cite[Exercise~14.1.1]{hilg11Ay}. Usually, the requirement for a linear nondegenerate mapping of $\mathfrak g$ to leave invariant each of the megaideals of $\mathfrak g$ produces enough constraints to understand the rough structure of automorphisms. Moreover, in~\cite{bihl11Cy,card12Ay,popo03Ay} several propositions were proved that enable the systematic computation of new megaideals using known ones. Given a system of differential equations~$\mathcal L$, let~$\mathfrak g_{\mathcal L}$ be its maximal Lie invariance algebra.
If a sufficient number of megaideals of the maximal Lie invariance algebra~$\mathfrak g_{\mathcal L}$ of the system~$\mathcal L$ is known, the condition of their invariance with respect to automorphisms of~$\mathfrak g_{\mathcal L}$ can be used in order to derive maximal restrictions on the form of a~general point symmetry transformation of~$\mathcal L$ that are obtainable within the framework of the algebraic method.
This is the essence of the megaideal-based version of the algebraic method or, shortly, the \emph{megaideal-based method}.

The effectiveness of this version has been demonstrated by computing the complete point symmetry group of the barotropic vorticity equation on the beta-plane~\cite{bihl11Cy}, the vorticity equation on the sphere~\cite{card12Ay}, the quasi-geostrophic two-layer model~\cite{bihl11By} and the system of primitive equations~\cite{card15Ay}. For each of these (systems of) differential equations, the corresponding maximal Lie invariance algebra is infinite dimensional.

The same techniques as applied to compute the complete point symmetry group of a system of differential equations can also be used to compute the complete equivalence group of a class of systems of differential equations. Moreover, the megaideal-based method is even more effective here than for symmetry groups as equivalence groups are often infinite dimensional. Showing how to realize the computation of equivalence groups using the algebraic method is the purpose of the present paper.
Both versions of the algebraic method are discussed.

We illustrate the megaideal-based method by the computation of the equivalence group for the class of nonlinear wave equations of the general form
\begin{subequations}\label{eq:IbragimovClass}
\begin{equation}\label{eq:IbragimovClassA}
u_{tt}=f(x,u_x)u_{xx}+g(x,u_x),
\end{equation}
with the nonvanishing condition
\begin{equation}\label{eq:IbragimovClassB}
(f_{u_x},g_{u_xu_x})\ne(0,0).
\end{equation}
\end{subequations}
This condition is essential as it guarantees that all equations from the class~\eqref{eq:IbragimovClass} are really nonlinear, cf.~\cite{bihl11Dy}. Physically, this class arises, in particular, in the study of nonlinear elasticity.

\looseness=1
The study of the class~\eqref{eq:IbragimovClass} plays a prominent role in the field of group classification of differential equations. It was initially considered in~\cite{ibra91Ay}, where one-dimensional symmetry extensions of the kernel algebra with respect to a fixed finite-dimensional subalgebra of the infinite-dimensional equivalence algebra of this class were found. In other words, a partial preliminary group classification problem was solved. A similar direction was taken in~\cite{hari93Ay}, where the partial preliminary group classification of the class~\eqref{eq:IbragimovClass} with respect to one-dimensional subalgebras of an infinite-dimensional subalgebra of the equivalence algebra was considered.
The group classification problem for systems of the form $v_t = a(x,v)w_x$ and $w_t = b(x,v)v_x$ related to equations from the original class~\eqref{eq:IbragimovClass} was first investigated in~\cite{ibra00Ay} and then exhaustively solved using the algebraic method of group classification in~\cite{khab09Ay}.
A study related to the group classification problem of the class~\eqref{eq:IbragimovClass} was carried out in~\cite{ibra04Ay}, resulting in the construction of second-order differential invariants of the equivalence algebra. Finally, the complete group classification problem for the class~\eqref{eq:IbragimovClass} was solved in~\cite{bihl11Dy} by the enhanced algebraic method of group classification, thus completing the over twenty years long investigation of the point symmetry properties of this~class.

The further structure of the present paper is as follows. In Section~\ref{sec:TheoryEquivalenceGroups},
we recall some definitions related to classes of systems of differential equations, their equivalence algebras and their equivalence groups.
Section~\ref{sec:TheoryAlgebraicMethod} contains the necessary material explaining the algebraic method for the computation of (complete point) equivalence groups.
Both versions of the method---automorphism- and megaideal-based---are introduced and presented in the form of step-by-step procedures.
In Section~\ref{sec:ExampleNonlinearWave}, we illustrate the megaideal-based method with the computation of the equivalence group of the class~\eqref{eq:IbragimovClass}.
In the final Section~\ref{sec:ConclusionAlgebraicMethod}, we summarize our results and give some thoughts about future directions of investigation.

\section{The equivalence group and equivalence algebra}\label{sec:TheoryEquivalenceGroups}

In this section we recall the central definition of the equivalence group and the equivalence algebra of a class of (systems of) differential equations.
These are two of the fundamental concepts that play a key role in group classification of differential equations.
As such, they are discussed extensively in the papers and books on this subject, see, e.g.,~\cite{akha91Ay,bihl11Dy,card11Ay,ovsi82Ay,popo06Ay,popo10Cy,popo10Ay}.
It is natural to first rigorously define the notion of a class of differential equations.

Let $\mathcal L_\theta$: \smash{$L(x,u_{(p)},\theta_{(q)}(x,u_{(p)}))=0$} denote a system of differential equations parameterized by the tuple of arbitrary elements $\theta(x,u_{(p)}) = (\theta^1(x,u_{(p)}),\dots,\theta^k(x,u_{(p)}))$. Here and in the following, the tuple $x=(x_1,\dots, x_n)$ consists of the independent variables and~$u_{(p)}$ includes both the tuple of dependent variables~$u=(u^1,\dots,u^m)$ as well as all the derivatives of~$u$ with respect to~$x$ of order up to~$p$.
By~$\theta_{(q)}$ we denote the partial derivatives of the arbitrary elements $\theta$ of order not exceeding $q$ for which both $x$ and $u_{(p)}$ act as the independent variables.

In the definition of the class~\eqref{eq:IbragimovClass} we explicitly included the inequality $(f_{u_x},g_{u_xu_x})\ne(0,0)$,
which restricts the values that the two arbitrary elements $f$ and $g$ of the class~\eqref{eq:IbragimovClass} can take.
Such a set of so-called auxiliary conditions on the arbitrary elements~$\theta$
is a necessary and important component of the precise definition of the class of differential equations
but it is often posed in an implicit way.
The procedure can be formalized as follows.

Consider the system of auxiliary differential equations $S(x,u_{(p)},\theta_{(q')}(x,u_{(p)}))=0$ and inequalities $\Sigma(x,u_{(p)},\theta_{(q')}(x,u_{(p)}))\ne0$,%
\footnote{%
Other kinds of inequalities ($>$, $<$, etc.) as well as collections of them are also relevant.
}
in which again both~$x$ and~$u_{(p)}$ act as independent variables. Components of the tuples $S$ and~$\Sigma$ are smooth functions of~$x$, $u_{(p)}$ and $\theta_{(q')}$. It is thus required that the tuple of arbitrary elements~$\theta$ runs through the solution set, denoted by~$\mathcal S$, of both the auxiliary equations $S=0$ and inequalities $\Sigma\ne0$.

The above discussion can be summarized in the following definition:

\begin{definition}
The set $\{\mathcal L_\theta\mid\theta\in\mathcal S\}$ denoted by~$\mathcal L|_{\mathcal S}$ is called a \emph{class of differential equations}
defined by parameterized systems~$\mathcal L_\theta$ and the set~$\mathcal S$ of arbitrary elements~$\theta$.
\end{definition}

In order to define the equivalence group of the class~$\mathcal L|_{\mathcal S}$ it is convenient to first introduce the notion of admissible transformations and the notion of equivalence groupoid~\cite{bihl11Dy,popo06Ay,popo10Cy,popo10Ay}.
Given the arbitrary elements $\theta,\tilde\theta\in\mathcal S$ with the associated systems $\mathcal L_\theta$ and \smash{$\mathcal L_{\tilde\theta}$} from the class~$\mathcal L|_{\mathcal S}$, the set $\mathrm T(\theta,\tilde \theta)$ of point transformations that map the system $\mathcal L_\theta$ to the system \smash{$\mathcal L_{\tilde\theta}$} is called the set of admissible (point) transformations from $\mathcal L_\theta$ to \smash{$\mathcal L_{\tilde\theta}$}.
In particular, the maximal point symmetry (pseudo)group $G_\theta$ of the system $\mathcal L_\theta$ coincides, by definition, with the set $\mathrm{T}(\theta,\theta)$.

A triplet $(\theta,\tilde\theta,\phi)$, where $\theta,\tilde\theta\in\mathcal S$ with $\mathrm T(\theta,\tilde\theta)\ne\varnothing$ and $\phi\in \mathrm T(\theta,\tilde\theta)$, is called an \emph{admissible transformation} in the class~$\mathcal L|_{\mathcal S}$.
In other words, an admissible transformation is a triplet consisting of
a pair of similar systems (a source and a target system) and a point transformation connecting these two systems.

The partial binary operation of composition is naturally defined for pairs of admissible transformations
for which the target system of the first admissible transformation and the source equation of the second admissible transformation coincides,
$(\theta,\tilde\theta,\phi)\circ(\tilde\theta,\bar\theta,\tilde\phi)=(\theta,\bar\theta,\tilde\phi\circ\phi)$.
Each admissible transformation is invertible, $(\theta,\tilde\theta,\phi)^{-1}=(\tilde\theta,\theta,\phi^{-1})$,
i.e., the inversion of admissible transformations is a unitary operation defined everywhere.
The set of admissible transformations of the class~$\mathcal L|_{\mathcal S}$ that is endowed with the operations of composition and taking the inverse
is called the \emph{equivalence groupoid} of this class and denoted by $\mathcal G^\sim=\mathcal G^\sim(\mathcal L|_{\mathcal S})$.
All the groupoid axioms are obviously satisfied for~$\mathcal G^\sim$.

For a general class of differential equations,
it is not possible to relate its equivalence groupoid with a group of point transformations
that act in the extended space of $(x,u_{(p)},\theta)$, respect the contact structure on the space of $(x,u_{(p)})$ and preserve the class.
However, it makes sense to single out a maximal part of the equivalence groupoid that admits such a relation;
then the associated group is referred to as the equivalence group.

\begin{definition}\label{def:UsualEquivalenceGroup}
The \emph{(usual) equivalence group}, denoted by $G^{\sim}=G^{\sim}(\mathcal L|_{\mathcal S})$, of the class~$\mathcal L|_{\mathcal S}$ is the (pseudo)group of point transformations in the space of $(x,u_{(p)},\theta)$, each element~$\Phi$ of which satisfies the following properties:
It is projectable to the space of $(x,u_{(p')})$ for any $p'$ with $0\leqslant p'\leqslant p$.
The projection $\Phi|_{(x,u_{(p')})}$ is the $p'$th order prolongation of $\Phi|_{(x,u)}$.
For any~$\theta$ from~$\mathcal S$ its image $\Phi\theta$ also belongs to~$\mathcal S$.
Finally, $\Phi|_{(x,u)}\in\mathrm T(\theta,\Phi\theta)$.
\end{definition}

Here we say that a point transformation~$\varphi$: $\tilde z=\varphi(z)$ in the space of variables $z=(z_1,\ldots,z_k)$
is projectable on the space of variables $z'=(z_{i_1},\ldots,z_{i_{k'}})$ with $1\leqslant i_1<\cdots<i_{k'}\leqslant k$
if the expressions for the transformed variables~$\tilde z'$ depend only on~$z'$.
The projection of~$\varphi$ to the $z'$-space is denoted by $\varphi|_{z'}$, $\tilde z'=\varphi|_{z'}(z')$.

Elements of $G^{\sim}$ are called \emph{equivalence transformations} of the class~$\mathcal L|_{\mathcal S}$.
Each equivalence transformation~$\Phi$ induces a family of admissible transformations parameterized by the arbitrary elements,
$\{(\theta,\Phi\theta,\Phi|_{(x,u)})\mid \theta\in\mathcal S\}$.
If the entire equivalence groupoid~$\mathcal G^\sim$ is induced by the equivalence group~$G^\sim$ in the above way,
then transformational properties of the class~$\mathcal L|_{\mathcal S}$ are particularly nice and
the class is called \emph{normalized} \cite{bihl11Dy,popo06Ay,popo10Cy,popo10Ay}.
The notion of normalized classes of differential equations serves as basis for developing the algebraic method of group classification of differential equations.

If the arbitrary elements~$\theta$ do not depend on derivatives of~$u$ of order greater than~$p'$, $p'<p$,
then the group~$G^{\sim}$ can be assumed to act in the space of $(x,u_{(p')},\theta)$.

In some cases, e.g., if the arbitrary elements~$\theta$ depend on $x$ and $u$ only,
we can neglect the condition that the transformation components for~$(x,u)$ of equivalence transformations do not involve~$\theta$,
which gives the \emph{generalized equivalence group} $G^{\sim}_{\rm gen}=G^{\sim}_{\rm gen}(\mathcal L|_{\mathcal S})$
of the class~$\mathcal L|_{\mathcal S}$ \cite{mele94Ay,popo06Ay,popo10Ay}.
Each element~$\Phi$ of~$G^{\sim}_{\rm gen}$ is a point transformation in the $(x,u,\theta)$-space
such that for any~$\theta$ from~$\mathcal S$ its image $\Phi\theta$ also belongs to~$\mathcal S$,
and $\Phi(\cdot,\cdot,\theta(\cdot,\cdot))|_{(x,u)}\in\mathrm{T}(\theta,\Phi\theta)$.

It was mentioned in the introduction that the computation of point transformations of systems of differential equations using the direct method usually involves the solution of nonlinear systems of differential equations. Thus, rather than studying finite equivalence transformations, it is common in the framework of group analysis of differential equations to only consider their infinitesimal counterparts. This restriction linearizes the nonlinear determining equations for finding such transformations, making their computation essentially algorithmic. Thus, to the continuous component of the equivalence group $G^\sim$ one may associate the Lie algebra $\mathfrak g^\sim$ of vector fields in the space of $(x,u_{(p)},\theta)$, which for any $0\leqslant p'\leqslant p$ are projectable to the space of $(x,u_{(p')})$ with the property that their projections to the space of $(x,u_{(p')})$ are the $p'$th order prolongations of their projections to the space of $(x,u)$. Moreover, these vector fields are the generators of one-parameter groups of equivalence transformations of the class $\mathcal L|_{\mathcal S}$.
The Lie algebra~$\mathfrak g^\sim$ is called the \emph{(usual) equivalence algebra} of the class $\mathcal L|_{\mathcal S}$.

In a similar way, by considering the generalized equivalence group of the class $\mathcal L|_{\mathcal S}$ instead of the usual one
and  by allowing the components of vector fields that correspond to $x$ and $u$ to also depend on $\theta$,
one can define the \emph{generalized equivalence algebra} of the class $\mathcal L|_{\mathcal S}$.

The (usual or generalized) equivalence algebra $\mathfrak g^\sim$ of a class~$\mathcal L|_{\mathcal S}$ can be found by invoking an infinitesimal invariance criterion completely analogous to that used for the computation of the infinitesimal generators of one-parameter point symmetry groups of a single system of differential equations~\cite{akha91Ay,ovsi82Ay}.

In order to use the infinitesimal invariance criterion properly, one has to consider a vector field~$Q$ on the extended space of independent variables~$x$, derivatives $u_{(p)}$ and arbitrary elements~$\theta$. The components of~$Q$ corresponding to derivatives of $u$ with respect to $x$ should respect the contact structure of the space of independent and dependent variables. In other words, these coefficients should be obtained using the general prolongation formula~\cite{olve86Ay,ovsi82Ay}. The vector field~$Q$ should be prolonged to the derivatives of the arbitrary elements assuming both $x$ and $u_{(p)}$ as the independent variables; again, the general prolongation formula should be used for this purpose.
The prolonged vector field is then applied to the joint system~$\Delta$ of the general form \smash{$L(x,u_{(p)},\theta_{(q)})=0$} of systems from the class~$\mathcal L|_{\mathcal S}$ and of the auxiliary conditions defining the set~$\mathcal S$ of values of the arbitrary elements~$\theta$. If necessary, differential consequences of equations from the system~$\Delta$  have to be considered. The remaining computational procedure then parallels the computation of infinitesimal generators of Lie point symmetries and involves the solution of an overdetermined linear system of partial differential equations.

\section{Algebraic method}\label{sec:TheoryAlgebraicMethod}

As was mentioned in the introduction, any point symmetry transformation~$\mathcal T$ of a system of differential equations~$\mathcal L$ generates an automorphism of the maximal Lie invariance algebra of~$\mathcal L$ via push-forwarding of vector fields in the space of system variables. This condition implies constraints for the transformation~$\mathcal T$ which are then taken into account in further calculations using the direct method~\cite{bihl11Cy,card12Ay,card15Ay,gray13Ay,hydo98Ay,hydo00By,hydo00Ay}.
The set of transformations found in the way described constitute the complete point symmetry group of the system~$\mathcal L$ including both continuous and discrete point transformations.

The above algebraic method can be easily extended to the framework of equivalence transformations. A basis for this is given by the following simple theorem, which is proved in the same way as the similar assertion for point symmetries.

\begin{theorem}\label{pro:BasicPropositionOnAlgMethodForFindingEquivGroup}
Let $\mathcal L|_{\mathcal S}$ be a class of (systems of) differential equations,
$G^\sim$ and~$\mathfrak g^\sim$ the equivalence group and the equivalence algebra of this class (of the same type, namely, either usual or generalized ones).%
\footnote{%
The further consideration does not depend on what kind of equivalence groups and equivalence algebras (usual or generalized) are involved.
For consistency, we should just handle the usual (resp.\ generalized) equivalence algebra when looking for the usual (resp.\ generalized) equivalence group
and impose relevant a priori restrictions on the admitted structure of equivalence transformations.
}
Any transformation~$\mathcal T$ from~$G^\sim$ induces an automorphism of~$\mathfrak g^\sim$ via push-forwarding of vector fields
in the relevant space of independent variables, derivatives of unknown functions and arbitrary elements of the class.
\end{theorem}

\begin{proof}
Consider an arbitrary vector field $Q\in\mathfrak g^\sim$.
The local one-parameter transformation group $G_Q=\{\exp(\ve Q)\}$ associated with $Q$
is contained in~$G^\sim$.
Then the one-parameter transformation group $\tilde G_Q=\{\mathcal T\exp(\ve Q)\mathcal T^{-1}\}$,
which is similar to~$G_Q$ with respect to~$\mathcal T$ and whose generator is $\mathcal T_*Q$,
is also contained in~$G^\sim$.
This means that the vector field $\mathcal T_*Q$ belongs to~$\mathfrak g^\sim$.
An arbitrary push-forward respects the Lie bracket of vector fields,
$[\mathcal T_*Q,\mathcal T_*Q']=\mathcal T_*[Q,Q']$ for any $Q,Q'\in\mathfrak g^\sim$.
Therefore, $\mathcal T_*$ is an automorphisms of~$\mathfrak g^\sim$.
\end{proof}

In other words, for each element~$\mathcal T$ of~$G^\sim$ there exists an automorphism~$\mathfrak T$ of~$\mathfrak g^\sim$
such that the condition $\mathcal T_*Q=\mathfrak T Q$ is satisfied for any $Q$ from~$\mathfrak g^\sim$.

It should be noted that while the correspondence $\mathcal T\to\mathcal T_*$ defines a representation of~$G^\sim$ in~$\mathrm{Aut}(\mathfrak g^\sim)$,
this representation is often unfaithful.
This is always the case if the equivalence group~$G^\sim$ (resp.\ the equivalence algebra~$\mathfrak g^\sim$) has a nontrivial center.
Moreover, the representation image~$G^\sim_*$ is a subgroup of the automorphism group $\mathrm{Aut}(\mathfrak g^\sim)$
that may be essentially smaller than the entire automorphism group $\mathrm{Aut}(\mathfrak g^\sim)$.
In other words, there may exist many automorphisms of~$\mathfrak g^\sim$ that are not induced by elements of~$G^\sim$.

The discussion on continuous point symmetries of a single system of differential equations~$\mathcal L$ given in~\cite{card12Ay} can be adopted to equivalence transformations as well. In particular, continuous equivalence transformations of the class~$\mathcal L|_{\mathcal S}$ can be composed and thus generate a connected normal subgroup, denoted by $U^\sim$, of the equivalence group~$G^\sim$. The transformations from $U^\sim$ induce mappings on the equivalence algebra~$\mathfrak g^\sim$ that are internal automorphisms of~$\mathfrak g^\sim$ thus generating the normal subgroup $\mathrm{Int}(\mathfrak g^\sim)$ of the automorphism group $\mathrm{Aut}(\mathfrak g^\sim)$. So, $\mathrm{Int}(\mathfrak g^\sim)=U^\sim_*$. The representatives of the elements of the factor group $G^\sim/U^\sim$ in~$G^\sim$ will be interpreted as \emph{discrete equivalence transformations} of the class~$\mathcal L|_{\mathcal S}$. It is clear from the construction that it makes sense to consider only discrete equivalence transformations that are independent up to composing them with continuous equivalence transformations of the class~$\mathcal L|_{\mathcal S}$. When $G^\sim$ is infinite dimensional, the factorization is usually difficult but it is sometimes still possible to realize it and thus to identify the discrete equivalence transformations. An example will be given in Section~\ref{sec:ExampleNonlinearWave}.

\begin{remark}
Classes of systems of differential equations with a single dependent variable (i.e., with $m=1$)
may possess contact equivalence transformations that are not prolongations of point transformations.
In this case the entire theory can be extended to contact equivalence transformations.
\end{remark}

\subsection{Version based on automorphisms}\label{sec:HydonMethod}

Theorem~\ref{pro:BasicPropositionOnAlgMethodForFindingEquivGroup} allows us to extend the automorphism-based method for finding point symmetry groups of systems of differential equations~\cite{hydo98Ay,hydo00By,hydo00Ay} to equivalence groups of classes of such systems. The extension properly works if the corresponding equivalence algebra~$\mathfrak g^\sim$ is of finite nonzero (and, moreover, low) dimension since the knowledge of the entire group $\mathrm{Aut}(\mathfrak g^\sim)$ is needed.
A new feature is that one should consider \emph{appropriate} point transformations and vector fields
in the relevant extended space of independent variables, derivatives of unknown functions and arbitrary elements, i.e., of $(x,u_{(p)},\theta)$.
If the arbitrary elements~$\theta$ do not depend on derivatives of~$u$ of order greater than~$p'$, $p'<p$,
then these appropriate objects can be restricted to the space of $(x,u_{(p')},\theta)$.
In the case $p'>0$ they should be consistent with the contact structure of the space of $(x,u_{(p')})$.
In the course of computing usual equivalence groups, they should also be projectable to the space of the system variables $(x,u)$.

We present the extension of the automorphism-based method in the form of a step-by-step procedure.
In what follows the indices $i$, $j$ and $k$ run from 1 to $n=\dim\mathfrak g^\sim$, and we assume summation over repeated indices.

\begin{enumerate}
\item
Let $\mathcal L|_{\mathcal S}$ be a class of (systems of) differential equations, find the equivalence algebra~$\mathfrak g^\sim$ of this class by the infinitesimal invariance criterion. Suppose that $n=\dim g^\sim<\infty$.
\item
Fix a basis $\{Q_1,\dots,Q_n\}$ of~$\mathfrak g^\sim$,
and compute the structure constants~$c_{ij}^k$ of~$\mathfrak g^\sim$ in this basis,
$[Q_i,Q_j]=c_{ij}^kQ_k$.
\item
The general form $(a^i_j)$ of automorphism matrices in the basis $\{Q_1,\dots,Q_n\}$ is obtained
via solving the system of algebraic equations \smash{$c_{i'j'}^{k'}a^{i'}_ia^{j'}_j=c_{ij}^ka^{k'}_k$}.%
\footnote{%
Automorphism groups have been computed for many finite-dimensional Lie algebras, in particular
for all semi-simple Lie algebras~\cite{jaco79Ay} and Lie algebras of dimension not greater than six,
see~\cite{chri03Ay,fish13Ay,gray13Ay,popo03Ay} and references therein.
The automorphism group of a decomposable Lie algebra is easily constructed from the automorphism groups of the decomposition components~\cite{fish13Ay}.
Finding the automorphism group of a Lie algebra~$\mathfrak g$ can be simplified
if the chosen basis of~$\mathfrak g$ is consistent with a known megaideal hierarchy of~$\mathfrak g$.
\vspace{.5ex}
}
\item
Take the general form of an appropriate point transformation $\mathcal T$ in the space of $x$, \smash{$u_{(p)}$} and $\theta$ and push-forward the vector fields $Q_1,\dots,Q_n$.
Then we set \smash{$\mathcal T_*Q_i=a^j_iQ_j$}, $i=1,\dots,n$, where \smash{$(a^j_i)$} is the general form of automorphism matrices.
Consequently equating the corresponding vector-field components of right- and left-hand sides of these equations
produces a system of differential equations for components of the transformation~$\mathcal T$.%
\footnote{%
The system derived involves parameters of the automorphism group $\mathrm{Aut}(\mathfrak g^\sim)$.
Therefore, in the course of integrating the system one in fact needs to solve a compatibility problem
in order to find values of the parameters for which the system is consistent with respect to components of~$\mathcal T$.
If the group $G^\sim_*$ does not coincide with $\mathrm{Aut}(\mathfrak g^\sim)$, which is often the case,
then there exist values of the parameters for which the system is inconsistent.
}
\item
Integrate the system and use the obtained intermediate form of $\mathcal T$ within the framework of the direct method in order to complete the system of constraints for $\mathcal T$ and to produce the final form for $\mathcal T$.
\end{enumerate}

Recall once more that the difference in the application of the above procedure to the computation of the usual (resp.\ generalized) equivalence group is that $\mathfrak g^\sim$ is the usual (resp.\ generalized) equivalence algebra; additionally, for the usual equivalence group one restricts oneself to transformations that are projectable to the space of $x$ and $u$ while for the generalized equivalence group one takes $\mathcal T$ to be a general point transformation in the space of $(x,u,\theta)$ without any additional a priori restrictions.

\begin{remark}
The procedure described above can be modified by factoring out the group~$U^\sim$ of continuous equivalence transformations of the class~$\mathcal L|_{\mathcal S}$,
which can easily be computed once the equivalence algebra~$\mathfrak g^\sim$ is known.
As $U^\sim_*=\mathrm{Int}(\mathfrak g^\sim)$,
in order to do the modification, instead of considering the general form of automorphism matrices in Step~3, one can factorize general automorphisms by inner automorphisms.
More precisely, suppose that for each equivalence class of automorphisms with respect to $\mathrm{Int}(\mathfrak g^\sim)$
one can take a representative in such a way that these representatives constitute a matrix group.
Then, this group is isomorphic to the factor group $\mathrm{Aut}(\mathfrak g^\sim)/\mathrm{Int}(\mathfrak g^\sim)$.
In general, the form of representatives is more restrictive than the general form of automorphism matrices,
which may simplify the equations for $\mathcal T$ derived in Step~3 of the above procedure.
The realization of the modified procedure results in the discrete equivalence transformations of the class~$\mathcal L|_{\mathcal S}$ composed with elements of the center of~$U^\sim$. \end{remark}

For an equivalence algebra of large finite dimension or, especially, for an infinite-dimensional equivalence algebra the computation of its automorphism groups may be a difficult problem, which also needs additional efforts for its rigorous formulation in the case of infinite dimension. Moreover, a significant part of automorphisms may not be realized via push-forwarding of vector fields by appropriate point transformations in the relevant space of independent and dependent variables, involved derivatives and arbitrary elements.
In such a case, the exhaustive description of automorphisms is not too principal
as then the equivalence transformations induce only a~(small) proper subgroup of the automorphism group.

\subsection{Version based on megaideals}

It has been pointed out in Section~\ref{sec:HydonMethod} that for the automorphism-based method, it is crucial that the equivalence algebra~$\mathfrak g^\sim$ is low dimensional. If this is not the case, the computation of the automorphism group $\mathrm{Aut}(\mathfrak g^\sim)$ becomes an intricate problem. It is a quite common situation for a class of differential equations to be studied that its equivalence algebra~$\mathfrak g^\sim$ is infinite dimensional and thus a more appropriate version of the algebraic method is needed. Mathematically, the main problem with the version reviewed in Section~\ref{sec:HydonMethod} is that in computing the automorphism group $\mathrm{Aut}(\mathfrak g^\sim)$ one also has to find the elements $\mathrm{Aut}(\mathfrak g^\sim)\setminus G^\sim_*$, which cannot be used for the computation of~$G^\sim_*$.

This deficiency is remedied by working with megaideals of~$\mathfrak g^\sim$ rather than with~$\mathrm{Aut}(\mathfrak g^\sim)$ itself
\cite{bihl11Cy,card12Ay,card15Ay}.
That is, we employ the fact that $G^\sim_*\mathfrak i\subseteq\mathfrak i$ for any megaideal $\mathfrak i$ of $\mathfrak g^\sim$,
rather than invoking the condition $G^\sim_*\subseteq\mathrm{Aut}(\mathfrak g^\sim)$.%
\footnote{%
Although the stronger condition $G^\sim_*\mathfrak i=\mathfrak i$ is always satisfied,
the inverse inclusion $G^\sim_*\mathfrak i\supseteq\mathfrak i$ is trivial in view of the presence of the identical automorphism in~$G^\sim_*$,
gives no constraints on elements of~$G^\sim$ and, therefore, can be neglected.
\vspace{1ex}
}

We recall that a \emph{fully characteristic ideal} (or, shortly, \emph{megaideal}) $\mathfrak i$ of a Lie algebra~$\mathfrak g$ is a vector subspace of $\mathfrak g$ that is invariant under any transformation from the automorphism group $\mathrm{Aut}(\mathfrak g)$ of~$\mathfrak g$ \cite{bihl11Cy,popo03Ay}, cf.\ \cite[Exercise~14.1.1]{hilg11Ay}.
We thus have that $\mathfrak T\mathfrak i=\mathfrak i$ for each megaideal~$\mathfrak i$ of~$\mathfrak g$,
whenever $\mathfrak T$ is a transformation from $\mathrm{Aut}(\mathfrak g)$.%
\footnote{%
The invariance condition $\mathfrak T^{-1}\mathfrak i\subseteq\mathfrak i$ of~$\mathfrak i$ with respect to $\mathfrak T^{-1}$
implies $\mathfrak i\subseteq\mathfrak T\mathfrak i$, which gives, together with the condition $\mathfrak T\mathfrak i\subseteq\mathfrak i$,
the equality $\mathfrak T\mathfrak i=\mathfrak i$.
}
Megaideals of~$\mathfrak g$ are ideals and, moreover, characteristic ideals of~$\mathfrak g$.
In the present context, the Lie algebra in this definition is the equivalence algebra~$\mathfrak g^\sim$ of the class of differential equations~$\mathcal L|_{\mathcal S}$.

In order to make the version of the megaideal-based method most effective, it is essential to be able to construct wide sets of megaideals. The more megaideals are known, the better restrictions on the form of admitted automorphisms can be derived. While it might be a complicated problem to obtain a complete list of megaideals of~$\mathfrak g^\sim$, for the practical application of the megaideal-based method the more tractable problem of finding a set of megaideals that gives maximal restrictions on point transformations should be tackled. Thus, it is for example not essential to consider megaideals that are sums of other megaideals, since they give weaker constraints than their individual summands jointly.

In~\cite{bihl11Cy,card12Ay,card15Ay,popo03Ay} it was shown how new megaideals can be computed from known ones. For the sake of completeness of the present paper, we collect the statements on megaideals from these papers. The central observation for practical computations is that many megaideals of a Lie algebra~$\mathfrak g$ can be constructed without prior knowledge of the automorphism group~$\mathrm{Aut}(\mathfrak g)$ using the following obvious assertions:

\begin{proposition}\label{pro:OnMegaIdeals}
 Let $\mathfrak i_1$ and $\mathfrak i_2$ be megaideals of a Lie algebra $\mathfrak g$. Then we have:
 \begin{enumerate}\itemsep=0ex
  \item The improper subalgebras of $\mathfrak g$ (i.e., the zero subspace and $\mathfrak g$ itself) are megaideals of $\mathfrak g$.
  \item The sum $\mathfrak i_1+\mathfrak i_2$, the intersection $\mathfrak i_1\cap \mathfrak i_2$ and the Lie product $[\mathfrak i_1,\mathfrak i_2]$ of megaideals are megaideals.
  \item If $\mathfrak i_2$ is a megaideal of $\mathfrak i_1$ and $\mathfrak i_1$ is a megaideal of $\mathfrak g$ then $\mathfrak i_2$ is a megaideal of $\mathfrak g$. Thus, megaideals of megaideals are again megaideals.
  \item The centralizer (resp.\ the normalizer) of a megaideal is a megaideal.
  \item All elements of the derived, upper and lower central series of a Lie algebra are its megaideals. It thus follows that the center and the derivative of~$\mathfrak g$ are megaideals.
  \item The radical~$\mathfrak r$ and nil-radical~$\mathfrak n$ (i.e., the maximal solvable and nilpotent ideals, respectively) of~$\mathfrak g$
as well as different Lie products, sums and intersections involving~$\mathfrak g$, $\mathfrak r$ and~$\mathfrak n$
($[\mathfrak g,\mathfrak r]$, $[\mathfrak r,\mathfrak r]$, $[\mathfrak g,\mathfrak n]$, $[\mathfrak r,\mathfrak n]$,  $[\mathfrak n,\mathfrak n]$, etc.) are megaideals of~$\mathfrak g$.
 \end{enumerate}
\end{proposition}

\noprint{
\begin{proof}
The assertions listed in Proposition~\ref{pro:OnMegaIdeals} are quite obvious.
Here we only prove the fourth assertion, which states
that the centralizer $\mathrm C_{\mathfrak g}(\mathfrak i)$ and the normalizer $\mathrm N_{\mathfrak g}(\mathfrak i)$ of a megaideal~$\mathfrak i$ in~$\mathfrak g$ are also megaideals of~$\mathfrak g$.

Consider arbitrary $\mathfrak T\in\mathrm{Aut}(\mathfrak g)$, $v\in\mathfrak i$, $w\in\mathrm C_{\mathfrak g}(\mathfrak i)$ and $u\in\mathrm N_{\mathfrak g}(\mathfrak i)$.
By the definitions of automorphism and centralizer we get
$[\mathfrak Tw,v]=[\mathfrak Tw,\mathfrak T\mathfrak T^{-1}v]=\mathfrak T[w,\mathfrak T^{-1}v]=0$
as $\mathfrak T^{-1}v\in\mathfrak i$ and hence $[w,\mathfrak T^{-1}v]=0$.
This means that $\mathfrak Tw\in\mathrm C_{\mathfrak g}(\mathfrak i)$.
Therefore, $\mathrm C_{\mathfrak g}(\mathfrak i)$ is a megaideal of~$\mathfrak g$.

Similarly, for the normalizer $\mathrm N_{\mathfrak g}(\mathfrak i)$ we have
$[\mathfrak Tu,v]=[\mathfrak Tu,\mathfrak T\mathfrak T^{-1}v]=\mathfrak T[u,\mathfrak T^{-1}v]\in\mathfrak i$
as $\mathfrak T^{-1}v\in\mathfrak i$ and hence $[u,\mathfrak T^{-1}v]\in\mathfrak i$.
Hence $\mathfrak Tu\in\mathrm N_{\mathfrak g}(\mathfrak i)$ and thus $\mathrm N_{\mathfrak g}(\mathfrak i)$ is a megaideal of~$\mathfrak g$.
\end{proof}
}

In order to find more megaideals without computing automorphisms,
we also apply an assertion,
which in general has no clear interpretation in terms of distinguished object related to the structure of the corresponding Lie algebra.
It was first proved in~\cite{card12Ay}, but for the sake of reference we repeat the proof here.

\begin{proposition}\label{pro:WayToFindMegaideals}
If~$\mathfrak i_0$, $\mathfrak i_1$ and~$\mathfrak i_2$ are megaideals of~$\mathfrak g$
then the set~$\mathfrak s$ of elements from~$\mathfrak i_0$ whose commutators with arbitrary elements from~$\mathfrak i_1$ belong to~$\mathfrak i_2$
is also a megaideal of~$\mathfrak g$.
\end{proposition}

\begin{proof}
It is clear that the set $\mathfrak s$ is a linear subspace of~$\mathfrak g$.
Consider an arbitrary element $z_0\in\mathfrak s$.
Hence $z_0\in\mathfrak i_0$ and $[z_0,z_1]\in\mathfrak i_2$ for any $z_1\in\mathfrak i_1$.
Then for any $\mathfrak T\in\mathrm{Aut}(\mathfrak g)$ and any $z_1\in\mathfrak i_1$ we have
$\mathfrak Tz_0\in\mathfrak i_0$ and $[\mathfrak Tz_0,z_1]=[\mathfrak Tz_0,\mathfrak T\mathfrak T^{-1}z_1]=\mathfrak T[z_0,\mathfrak T^{-1}z_1]\in\mathfrak i_2$
as $\mathfrak T^{-1}z_1\in\mathfrak i_1$ and thus $[z_0,\mathfrak T^{-1}z_1]\in\mathfrak i_2$.
This means that $\mathfrak Tz_0\in\mathfrak s$. Therefore, $\mathfrak s$ is a megaideal of~$\mathfrak g$.
\end{proof}

Assertion (iv) of Proposition~\ref{pro:OnMegaIdeals} is a particular case of Proposition~\ref{pro:WayToFindMegaideals},
where $\mathfrak i_0=\mathfrak g$ and $\mathfrak i_2=\{0\}$ for the centralizer of $\mathfrak i_1$,
and $\mathfrak i_0=\mathfrak g$ and $\mathfrak i_2=\mathfrak i_1$ for the normalizer of $\mathfrak i_1$.

\begin{remark}\label{rem:OneMoreWayToFindMegaideals}
If $\mathfrak m$ is a finite-dimensional megaideal and its automorphism group~$\mathrm{Aut}(\mathfrak m)$ is already known,
all megaideals of~$\mathfrak m$ can be found by direct computation according to the definition of megaideals
as subspaces of~$\mathfrak m$ invariant with respect to~$\mathrm{Aut}(\mathfrak m)$.
In the course of calculating the automorphisms we can use knowledge about structural megaideals of~$\mathfrak m$
such as the center, the radical, the nilradical, elements of the derived, lower central and upper central series of~$\mathfrak m$
and different megaideals related to structural megaideals via certain operations.
\end{remark}

Theorem~\ref{pro:BasicPropositionOnAlgMethodForFindingEquivGroup} implies that
any transformation~$\mathcal T$ from~$G^\sim$ satisfies the condition $\mathcal T_*\mathfrak i=\mathfrak i$ for each megaideal~$\mathfrak i$ of~$\mathfrak g^\sim$.
In the course of the derivation of constraints for $\mathcal T$, this condition is interpreted in the following way:
Let the megaideal $\mathfrak i$ consist of the vector fields $Q_\gamma$,
where $\gamma$ runs through a parameter set $\Gamma_{\mathfrak i}$,
and be spanned by the vector fields $Q_{\gamma'}$ parameterized by $\gamma'$ running through a subset $\Gamma'_{\mathfrak i}$ of $\Gamma_{\mathfrak i}$.
The condition $\mathcal T_*\mathfrak i=\mathfrak i$ implies that for any $\gamma'\in\Gamma'_{\mathfrak i}$,
there exists a $\gamma\in\Gamma_{\mathfrak i}$ such that $\mathcal T_*Q_{\gamma'}=Q_{\gamma}$.

Similarly to the automorphism-based method, the megaideal-based counterpart can also be split in a few algorithmic steps.
Some steps are the same or almost the same as for the automorphism-based method.
Even so, for the convenience of further application we describe each step in its entirety.
We also re-interpret the condition for megaideal invariance in a manner more convenient for practical use.

\begin{enumerate}
\item
Given a class of (systems of) differential equations~$\mathcal L|_{\mathcal S}$, find the equivalence algebra~$\mathfrak g^\sim$ of this class by the infinitesimal invariance criterion. Let~$\mathfrak g^\sim=\{Q_\gamma\mid\gamma\in\Gamma\}$ for some parameter set $\Gamma$.
\item
Fix a set $\{Q_{\gamma'}\mid\gamma'\in\Gamma'\}$ of vector fields spanning~$\mathfrak g^\sim$,
$\mathfrak g^\sim=\langle Q_{\gamma'}\mid\gamma'\in\Gamma'\rangle$ with $\Gamma'\subset\Gamma$,
and compute commutation relations for all pairs of spanning vector fields.
\item
Using Propositions~\ref{pro:OnMegaIdeals} and~\ref{pro:WayToFindMegaideals}, Remark~\ref{rem:OneMoreWayToFindMegaideals} and other tools,
construct as wide a list of megaideals of the equivalence algebra~$\mathfrak g^\sim$ as possible.
Optimize the list by the exclusion of inessential megaideals, which give no new constraints for automorphisms of~$\mathfrak g^\sim$
as compared with other megaideals. In particular, megaideals being sums of other megaideals are not essential.
\item
Take the general form of an appropriate point transformation $\mathcal T$ in the space of~$x$, $u_{(p)}$ and~$\theta$ and push-forward the vector fields $Q_{\gamma'}$, $\gamma'\in\Gamma'$.
For each $\gamma'$ we choose, from the above list, the minimal megaideal $\mathfrak i_{\gamma'}=\{Q_\gamma\mid\gamma\in\Gamma_{\mathfrak i_{\gamma'}}\subset\Gamma\}$
containing $Q_{\gamma'}$ and set $\mathcal T_*Q_{\gamma'}=Q_{\gamma_{\gamma'}}$,
where the parameter $\gamma_{\gamma'}$ satisfies the constraints singling out the set~$\Gamma_{\mathfrak i_{\gamma'}}$ from the entire set~$\Gamma$.
Consequently equating the corresponding vector-field components of right- and left-hand sides
in the equation \smash{$\mathcal T_*Q_{\gamma'}=Q_{\gamma_{\gamma'}}$} for each~$\gamma'$
leads to a system of differential equations for components of the transformation~$\mathcal T$.%
\footnote{%
Similarly to the automorphism-based method, the system obtained involves the parameters $\gamma_{\gamma'}$, $\gamma'\in\Gamma'$.
Therefore, integrating the system requires solving a compatibility problem
in order to find values of the parameters for which the system is consistent with respect to components of~$\mathcal T$.
In addition to that the group $G^\sim_*$ may be a proper subgroup of $\mathrm{Aut}(\mathfrak g^\sim)$,
for the megaideal-based version there may be one more reason for the system to be inconsistent for some parameter values:
In general, even the complete megaideal hierarchy cannot properly capture delicate constraints for elements of $\mathrm{Aut}(\mathfrak g^\sim)$.
}
The solution of this system results in an intermediate form for~$\mathcal T$.
\item
Use the derived form of~$\mathcal T$ within the framework of the direct method in order to complete the system of constraints for $\mathcal T$
and to produce the final form for $\mathcal T$ by solving the completed system.
\end{enumerate}

\section{Complete equivalence group of the class\\ of nonlinear wave equations}\label{sec:ExampleNonlinearWave}

In order to clarify all the algorithmic steps,
in this section we demonstrate in detail the computation of the usual equivalence group~$G^\sim$ of the class of nonlinear wave equations~\eqref{eq:IbragimovClass}
using the megaideal-based method.
The preliminary presentation of this result is contained in the arXiv preprint of~\cite{bihl11Dy}.

\subsection{Equivalence algebra}\label{sec:EquivalenceAlgebraIbragimovClass}

The equivalence algebra of class~\eqref{eq:IbragimovClass} including both linear and nonlinear equations was first computed in~\cite{ibra91Ay}. It was then shown in~\cite{bihl11Dy}, that the class consisting of only nonlinear equations admits the same equivalence algebra. Here we use both the representation of the equivalence algebra and the notation introduced in~\cite{bihl11Dy}. The equivalence algebra~$\mathfrak g^\sim$ of the class~\eqref{eq:IbragimovClass} is spanned by the vector fields
\begin{gather}\label{eq:EquivalenceAlgebraGenWaveEqs}
\begin{split}&
\DDD^u=u\p_u+u_x\p_{u_x}+g\p_g,\quad
\DDD^t=t\p_t-2f\p_f-2g\p_g,\quad
\PP^t=\p_t,\\&
\DDD(\varphi)=\varphi\p_x-\varphi_xu_x\p_{u_x}+2\varphi_xf\p_f+\varphi_{xx}u_xf\p_g,\\&
\GG(\psi)=\psi\p_u+\psi_x\p_{u_x}-\psi_{xx}f\p_g,\quad
\FF^1=t\p_u, \quad
\FF^2=t^2\p_u+2\p_g,
\end{split}
\end{gather}
where~$\varphi=\varphi(x)$ and~$\psi=\psi(x)$ run through the set of smooth functions of~$x$. It has been pointed out in~\cite{bihl11Dy} that while the elements of the equivalence algebra of class~\eqref{eq:IbragimovClass} are vector fields in the space of $(t,x,u_{(2)},f,g)$, for practical purposes it suffices to present only their projection to the space of $(t,x,u,u_x,f,g)$. The reason for this is that the vector field components corresponding to~$u_t$, $u_x$, $u_{tt}$, $u_{tx}$ and $u_{xx}$, can be obtained via prolongation from the components corresponding to $t$, $x$ and $u$. At the same time, we explicitly include the components associated with the derivative~$u_x$ in the representation of vector fields~\eqref{eq:EquivalenceAlgebraGenWaveEqs} spanning $\mathfrak g^\sim$ since the arbitrary elements~$f$ and~$g$ depend on $u_x$. Therefore, for the proper computation of commutation relations, the components with $\p_{u_x}$ are crucially needed.

For the construction of megaideals, it is convenient to also recall the nonvanishing commutation relations between the vector fields of~$\mathfrak g^\sim$. They read
\begin{align*}
 &[\GG(\psi),\DDD^u]=\GG(\psi), \quad [\FF^1,\DDD^u]=\FF^1, \quad [\FF^2,\DDD^u]=\FF^2,\\
 &[\DDD^t,\FF^1]=\FF^1, \quad [\DDD^t,\FF^2]=2\FF^2,\\
 &[\PP^t,\DDD^t]=\PP^t, \quad [\PP^t,\FF^1]=\GG(1), \quad [\PP^t,\FF^2]=2\FF^1, \\
 &[\DDD(\varphi^1),\DDD(\varphi^2)]=\DDD(\varphi^1\varphi^2_x-\varphi^1_x\varphi^2), \quad [\DDD(\varphi),\GG(\psi)]=\GG(\varphi\psi_x).
\end{align*}

Therefore, steps 1 and 2 of the megaideal-based method were in fact realized in~\cite{bihl11Dy,ibra91Ay}.

\subsection{Megaideals of equivalence algebra}
\label{sec:MegaidealsOfEquivalenceAlgebra}

In order to compute the complete equivalence group of the class~\eqref{eq:IbragimovClass} by the megaideal-based method
(cf. Section~\ref{sec:CalculationOfEquivAlgebraByAlgebraicMethod}),
we need to know a set of megaideals of the equivalence algebra~$\mathfrak g^\sim$ of this class.

Let $\mathfrak g=\mathfrak g^\sim$ for the sake of notational simplicity. Using Proposition~\ref{pro:OnMegaIdeals}, it is easy to compute the following megaideals of~$\mathfrak g^\sim$:
\begin{gather*}
\mathfrak g'=\langle\PP^t,\DDD(\varphi),\GG(\psi),\FF^1,\FF^2\rangle,\quad
\mathfrak g''=\langle\DDD(\varphi),\GG(\psi),\FF^1\rangle,\quad
\mathfrak g'''=\langle\DDD(\varphi),\GG(\psi)\rangle,\\
\mathrm C_{\mathfrak g}(\mathfrak g''')=\langle\DDD^t,\PP^t,\GG(1),\FF^1,\FF^2\rangle,\quad
\mathrm C_{\mathfrak g'}(\mathfrak g''')=\langle\PP^t,\GG(1),\FF^1,\FF^2\rangle,\\
\mathrm C_{\mathfrak g'}(\mathfrak g'')=\langle\GG(1),\FF^1,\FF^2\rangle,\quad
\mathrm Z_{\mathfrak g''}=\langle\GG(1),\FF^1\rangle,\quad
\mathrm Z_{\mathfrak g'}=\langle\GG(1)\rangle,\\
\mathrm R_{\mathfrak g}=\langle\DDD^u,\DDD^t,\PP^t,\GG(\psi),\FF^1,\FF^2\rangle,\quad
\mathrm R_{\mathfrak g'''}=\langle\GG(\psi)\rangle,
\end{gather*}
where $\mathfrak a'$, $\mathrm R_{\mathfrak a}$, $\mathrm Z_{\mathfrak a}$ and $\mathrm C_{\mathfrak a}(\mathfrak b)$
denote the derivative, the radical and the center of a Lie algebra~$\mathfrak a$
and the centralizer of a subalgebra~$\mathfrak b$ in~$\mathfrak a$, respectively.
We present proofs only for the last two equalities.

The linear span $\mathfrak s_1=\langle\DDD^u,\DDD^t,\PP^t,\GG(\psi),\FF^1,\FF^2\rangle$
obviously is a solvable ideal of~$\mathfrak g$.
Moreover, it is the maximal solvable ideal of~$\mathfrak g$.
Indeed, suppose that $\mathfrak s_1\subsetneq\mathfrak i$ and $\mathfrak i$ is an ideal of~$\mathfrak g$.
Then there exists a smooth function~$\zeta$ of~$x$ which does not identically vanish
such that the vector field $\DDD(\zeta)$ belongs to~$\mathfrak i$.
As $\mathfrak i$ is an ideal of~$\mathfrak g$,
for an arbitrary smooth function~$\varphi$ of~$x$ the commutator $[\DDD(\zeta),\DDD(\varphi)]=\DDD(\zeta\varphi_x-\zeta_x\varphi)$ belongs to~$\mathfrak i$.
If $\zeta$ is not a constant function, we define the following series of operators:
\[
R^{0k}=k^{-1}[\DDD(\zeta),\DDD(\zeta^{k+1})],\quad
R^{jk}=k^{-1}[R^{j-1,1},R^{j-1,k+1}],\quad
j,k=1,2,\dots.
\]
It is possible to prove by induction that $R^{j-1,k}=\DDD(\zeta^k(\zeta\zeta_x)^{2^j-1})\ne0$, $j,k=1,2,\dots$.
Moreover, as $R^{0k}\in\mathfrak i$, we have $R^{jk}\in\mathfrak i^{(j)}$, $j,k=1,2,\dots$, i.e., $\mathfrak i^{(j)}\ne\{0\}$ for any nonnegative integer~$j$.
This means that the ideal~$\mathfrak i$ is not solvable.
If $\zeta$ is a constant function, we can set $\zeta\equiv1$.
We choose any smooth function~$\varphi$ of~$x$ with $\varphi_{xx}\not\equiv0$ and denote~$\varphi_x$ by~$\tilde\zeta$.
As the commutator $[\DDD(1),\DDD(\varphi)]=\DDD(\tilde\zeta)$ belongs to~$\mathfrak i$,
the consideration for the previous case again implies that the ideal~$\mathfrak i$ is not solvable.
Therefore, $\mathfrak s_1$ is really the maximal solvable ideal of~$\mathfrak g$, i.e., $\mathrm R_{\mathfrak g}=\mathfrak s_1$.

The linear span $\mathfrak s_2=\langle\GG(\psi)\rangle$ is an Abelian and, therefore, solvable ideal of~$\mathfrak g'''$.
The maximality of this solvable ideal is proved in the same way as for~$\mathfrak s_1$.
Hence $\mathrm R_{\mathfrak g'''}=\mathfrak s_2$.

The same megaideals can be obtained in different ways. For example, $\langle\GG(1)\rangle=\mathrm Z_{\mathfrak g'}=\mathrm Z_{\mathfrak g'''}$.

To find one more megaideal which will be used in the course of the computation of the complete equivalence group of the class~\eqref{eq:IbragimovClass} by the algebraic method
in Section~\ref{sec:CalculationOfEquivAlgebraByAlgebraicMethod},
we should apply a more sophisticated technique than before.
There are two ways to do this.

The  way directly based on the definition of megaideals is to calculate the automorphism group~$\mathrm{Aut}(\mathfrak m)$
of the finite-dimensional megaideal $\mathfrak m=\mathrm C_{\mathfrak g}(\mathfrak g''')=\langle\DDD^t,\PP^t,\GG(1),\FF^1,\FF^2\rangle$
and then determine megaideals of~$\mathfrak m$ as subspaces of~$\mathfrak m$ which are invariant with respect to~$\mathrm{Aut}(\mathfrak m)$.
In the course of calculating the automorphisms we use the knowledge about simple megaideals of~$\mathfrak m$,
\[
\mathfrak m'=\langle\PP^t,\GG(1),\FF^1,\FF^2\rangle,\quad
\mathfrak m''=\langle\GG(1),\FF^1\rangle,\quad
\mathrm Z_{\mathfrak m}=\langle\GG(1)\rangle,\quad
\mathrm C_{\mathfrak m}(\mathfrak m'')=\langle\GG(1),\FF^1,\FF^2\rangle.
\]
The presence of the above set of nested megaideals is equivalent to that for any automorphism~$A$ of~$\mathfrak m$,
its matrix $(a_{ij})_{i,j=1}^5$ in the basis $\{\GG(1),\FF^1,\FF^2,\PP^t,\DDD^t\}$ is upper triangular with nonzero diagonal elements.
In particular,
\begin{align*}
&A\PP^t=a_{14}\GG(1)+a_{24}\FF^1+a_{34}\FF^2+a_{44}\PP^t, \\
&A\DDD^t=a_{15}\GG(1)+a_{25}\FF^1+a_{35}\FF^2+a_{45}\PP^t+a_{55}\DDD^t,
\end{align*}
where $a_{44}a_{55}\ne0$.
As $[\PP^t,\DDD^t]=\PP^t$ and $A\in\mathrm{Aut}(\mathfrak m)$, we should have $[A\PP^t,A\DDD^t]=A\PP^t$.
Collecting coefficients of basis elements in the last equality, we derive a system of equations with respect to~$a$'s
which implies, in view of the condition $a_{44}\ne0$, that $a_{55}=1$, $a_{34}=0$, $a_{24}=a_{44}a_{35}$ and
$a_{14}=a_{44}a_{25}-a_{45}a_{24}$.
As $a_{34}=0$, we get that the span $\langle\PP^t,\GG(1),\FF^1\rangle$ is a megaideal of~$\mathfrak m$ and, therefore, of~$\mathfrak g^\sim$.

The other way for finding the megaideal $\langle\PP^t,\GG(1),\FF^1\rangle$, which is based on Proposition~\ref{pro:WayToFindMegaideals},
allows us to avoid the above computation of automorphisms of the megaideal~$\mathfrak m$ and, thus, is much simpler:
Choosing $\mathfrak i_0=\mathfrak i_1=\mathrm C_{\mathfrak g'}(\mathfrak g''')=\langle\PP^t,\GG(1),\FF^1,\FF^2\rangle$
and $\mathfrak i_2=\mathrm Z_{\mathfrak g'}=\langle\GG(1)\rangle$ in Proposition~\ref{pro:WayToFindMegaideals},
we obtain the megaideal $\mathfrak s=\langle\PP^t,\GG(1),\FF^1\rangle$.

\subsection{Calculation of equivalence group}
\label{sec:CalculationOfEquivAlgebraByAlgebraicMethod}

The careful study of megaideals of the equivalence algebra~$\mathfrak g^\sim$ in Section~\ref{sec:MegaidealsOfEquivalenceAlgebra}
according to the third step of the megaideal-based procedure
supplies us with a sufficient store of megaideals in order to commence the computations directly concerned with equivalence transformations.

\begin{theorem}\label{thm:EquivalenceGroupIbragrimovClass}
The equivalence group~$G^\sim$ of the class~\eqref{eq:IbragimovClass} consists of the transformations
\begin{align}\label{eq:EquivalenceGroupIbragimovClass}
\begin{split}
 &\tilde t = c_1t+c_0, \quad \tilde x=\varphi(x), \quad \tilde u = c_2u+c_4t^2+c_3t+\psi(x), \quad \tilde u_{\tilde x}=\frac{c_2u_x+\psi_x}{\varphi_x},\\
 &\tilde f = \frac{\varphi_x^2}{c_1^2}f, \quad \tilde g = \frac{1}{c_1^2}\left(c_2g+\frac{c_2u_x+\psi_x}{\varphi_x}\varphi_{xx}f-\psi_{xx}f+2c_4\right),
\end{split}
\end{align}
where $c_0$, \dots, $c_4$ are arbitrary constants satisfying the condition $c_1c_2\ne0$ and $\varphi$ and~$\psi$ run through the set of smooth functions of~$x$, $\varphi_x\ne0$.
\end{theorem}

\begin{proof}
The group~$G^\sim$ consists of nondegenerate point transformations
in the joint space of variables~$t$, $x$ and $u$,
the first derivatives~$u_t$ and~$u_x$ and the arbitrary elements~$f$ and~$g$,
which are projectable to the variable space and
whose components for first derivatives are defined via
the first prolongation of their projections to the variable space.
Thus, the general form of a transformation~$\mathcal T$ from~$G^\sim$ is
\begin{gather*}
\tilde t=T(t,x,u),\quad \tilde x=X(t,x,u),\quad \tilde u=U(t,x,u),\\
\tilde u_{\tilde t}=U^t(t,x,u,u_t,u_x),\quad \tilde u_{\tilde x}=U^x(t,x,u,u_t,u_x),\\
\tilde f=F(t,x,u,u_t,u_x,f,g),\quad \tilde g=G(t,x,u,u_t,u_x,f,g),
\end{gather*}
where~$U^t$ and~$U^x$ are determined via~$T$, $X$ and~$U$ and
the nondegeneracy condition should be additionally satisfied.
To obtain the constrained form of $\mathcal T$,
we will act by the push-forward~$\mathcal T_*$ induced by~$\mathcal T$
on the vector fields~\eqref{eq:EquivalenceAlgebraGenWaveEqs} additionally including the terms with~$\p_{u_t}$
and use megaideals of the equivalence algebra~$\mathfrak g^\sim$ of the class~\eqref{eq:IbragimovClass}
and restrictions on automorphisms of~$\mathfrak g^\sim$ found in Section~\ref{sec:MegaidealsOfEquivalenceAlgebra}.
The principal part of the consideration is based on objects and properties of~$\mathfrak g^\sim$
that are related to the finite-dimensional megaideal
$\mathfrak m=\mathrm C_{\mathfrak g}(\mathfrak g''')=\langle\DDD^t,\PP^t,\GG(1),\FF^1,\FF^2\rangle$.
Recall that megaideals being sums of other megaideals are not essential for the computation
since they give weaker constraints than their summands.
For example, the megaideal~$\mathfrak g''$ is the sum of~$\mathfrak g'''$ and~$\mathrm Z_{\mathfrak g'}$
and hence we do not use it in the further consideration.
A list of essential megaideals is in fact exhausted by the spans
\begin{gather*}
\langle\GG(1)\rangle,\quad
\langle\GG(1),\FF^1\rangle,\quad
\langle\GG(1),\FF^1,\FF^2\rangle,\quad
\langle\PP^t,\GG(1),\FF^1\rangle,\quad
\langle\DDD^t,\PP^t,\GG(1),\FF^1,\FF^2\rangle,\\
\langle\GG(\psi)\rangle,\quad
\langle\DDD(\varphi),\GG(\psi)\rangle.
\end{gather*}
We apply the condition of preserving megaideals following their ordering by inclusion.
In other words, for each vector field~$Q$ from~\eqref{eq:EquivalenceAlgebraGenWaveEqs}
we take the megaideal that contains~$Q$ and is minimal among the above listed megaideals.
As a result, we derive an optimized set of constraints for $\mathcal T_*$ as automorphism of~$\mathfrak g$,
\begin{subequations}\label{eq:PushForwardOfT}
\begin{align}
\label{eq:PushForwardOfT1}
&\mathcal T_*\GG(1)=a_{11}\tilde\GG(1),
\\\label{eq:PushForwardOfT2}
&\mathcal T_*\FF^1=a_{12}\tilde\GG(1)+a_{22}\tilde\FF^1,
\\\label{eq:PushForwardOfT3}
&\mathcal T_*\FF^2=a_{13}\tilde\GG(1)+a_{23}\tilde\FF^1+a_{33}\tilde\FF^2,
\\\label{eq:PushForwardOfT4}
&\mathcal T_*\PP^t=a_{14}\tilde\GG(1)+a_{24}\tilde\FF^1+a_{44}\tilde\PP^t,
\\\label{eq:PushForwardOfT5}
&\mathcal T_*\DDD^t=a_{15}\tilde\GG(1)+a_{25}\tilde\FF^1+a_{35}\tilde\FF^2+a_{45}\tilde\PP^t+\tilde\DDD^t,
\\\label{eq:PushForwardOfT6}
&\mathcal T_*\GG(\hat\psi)=\tilde\GG(\tilde\psi^{\hat\psi}),
\\\label{eq:PushForwardOfT7}
&\mathcal T_*\DDD(\hat\varphi)=\tilde\GG(\tilde\psi^{\hat\varphi})+\tilde\DDD(\tilde\varphi^{\hat\varphi}),
\end{align}
\end{subequations}
where $a$'s are constants, $a_{11}a_{22}a_{33}a_{44}\ne0$ and
$\hat\psi$ and $\hat\varphi$ are arbitrary smooth functions of~$x$.
The constants~$a$'s completed with $a_{55}=1$ and $a_{ij}=0$, $1\leqslant i<j\leqslant5$,
form a matrix of an automorphism of the megaideal~$\mathfrak m$.
Tildes over vector fields on the right hand sides of the above equations mean
that these vector fields are written in terms of the transformed variables.
The parameter-functions \smash{$\tilde\psi^{\hat\psi}$}, $\tilde\psi^{\hat\varphi}$ and $\tilde\varphi^{\hat\varphi}$
are smooth functions of~$\tilde x$ associated with the parameter-functions~\smash{$\hat\psi$} or $\hat\varphi$,
which is indicated by the corresponding superscripts.
We will derive constraints for~$\mathcal T$,
consequently equating the corresponding vector-field components of right- and left-hand sides in each of the conditions~\eqref{eq:PushForwardOfT}
and taking into account constraints obtained in previous steps.
As the components of vector fields and of the transformation~$\mathcal T$
associated with the derivatives $u_t$ and~$u_x$ are defined via first-order prolongation
involving the similar values related to the variables~$t$, $x$ and~$u$,
the coefficients of~$\p_{u_t}$ and~$\p_{u_x}$ give no essentially new equations
in comparison with the coefficients of~$\p_t$, $\p_x$ and~$\p_u$.
This is why we will not equate the coefficients of~$\p_{u_t}$ and~$\p_{u_x}$.
To have a proper representation of the final result, we will re-denote certain values in an appropriate way.

Thus, the equation~\eqref{eq:PushForwardOfT1} implies that $T_u=X_u=0$, $U_u=c_2$ and $F_u=G_u=0$,
where the nonvanishing constant~$a_{11}$ is re-denoted by~$c_2$.
Then we derive from Eq.~\eqref{eq:PushForwardOfT2} that
$tU_u=a_{22}T+a_{12}$, i.e., $T=c_1t+c_0$ where $c_1=c_2/a_{22}\ne0$ and $c_0=-a_{12}/a_{22}$, and $F_{u_t}=G_{u_t}=0$.
The consequence $t^2U_u=a_{33}T^2+a_{23}T+a_{13}$ of Eq.~\eqref{eq:PushForwardOfT3} gives only relations between $a$'s.
In particular, $a_{33}=c_2/c_1^2$.
Then the other consequences of Eq.~\eqref{eq:PushForwardOfT3} are $F_g=0$ and~$G_g=c_2/c_1^2$.
The essential consequences of Eq.~\eqref{eq:PushForwardOfT4} are exhausted by $X_t=0$, $U_t=a_{24}T+a_{14}$ and $F_t=G_t=0$.
Therefore, $X=\varphi(x)$ and $U=c_2+c_4t^2+c_3t+\psi(x)$, where $\varphi_x\ne0$, $c_4=a_{24}c_1/2$ and $c_3=a_{14}+a_{24}c_0$.

As we have already derived the precise expressions for the components of~$\mathcal T$ corresponding to the variables
(cf.\ Eq.~\eqref{eq:EquivalenceGroupIbragimovClass}),
at this point we could terminate the computation of equivalence transformations by the algebraic method
and calculate the expressions for~$F$ and~$G$ by the direct method.
At the same time, all the determining equations for transformations from the equivalence group~$G^\sim$
of the class~\eqref{eq:IbragimovClass} follow from restrictions for automorphisms of the equivalence algebra~$\mathfrak g^\sim$.
This is not a common situation when the algebraic method is applied.
Usually it gives only a part of the determining equations simplifying the subsequent application of the direct method.
See, e.g., the computations of the complete point symmetry groups of
the barotropic vorticity equation and quasi-geostrophic two-layer model
in \cite[Section~3]{bihl11Cy} and \mbox{\cite[Section~4]{bihl11By}}, respectively.
This is why we complete the consideration of the equivalence group~$G^\sim$ within the framework of the algebraic method.

From Eq.~\eqref{eq:PushForwardOfT5} we obtain in particular that $tU_t=a_{35}T^2+a_{25}T+a_{15}$, $fF_f=F$ and $fG_f+gG_g=G-a_{35}$,
where $a_{35}=2c_4/c_1^2$ in view of the first of these consequences.

The equation~\eqref{eq:PushForwardOfT6} implies the equations
\begin{equation}\label{eq:ConditionsForArbitraryElementsByAlgabraicMethod1}
\hat\psi U_u=\tilde\psi^{\hat\psi},\quad
\hat\psi_xF_{u_x}=0,\quad
\hat\psi_xG_{u_x}-\hat\psi_{xx}fG_g=\tilde\psi^{\hat\psi}_{\tilde x\tilde x}F.
\end{equation}
The first and second equations of~\eqref{eq:ConditionsForArbitraryElementsByAlgabraicMethod1} are equivalent to
$\tilde\psi^{\hat\psi}=c_2\hat\psi$ and $F_{u_x}=0$.
Then we can express the derivative $\tilde\psi^{\hat\psi}_{\tilde x\tilde x}$ via derivatives of~$\hat\psi$,
$\tilde\psi^{\hat\psi}_{\tilde x\tilde x}=c_2\varphi_x^{-3}(\varphi_x\hat\psi_{xx}-\varphi_{xx}\hat\psi_x)$,
substitute the expression into the third equation of~\eqref{eq:ConditionsForArbitraryElementsByAlgabraicMethod1}
and split with respect to the derivatives $\hat\psi_x$ and~$\hat\psi_{xx}$,
as the function~$\hat\psi$ is arbitrary.
As a result, we obtain $F=c_1^{-2}\varphi_x{}^2f$ and
$G_{u_x}=c_2\varphi_x^{-3}\varphi_{xx}F$, i.e., $G_{u_x}=c_2c_1^{-2}\varphi_x^{-1}\varphi_{xx}f$.
The expression for~$F$ coincides with the transformation component for~$f$ presented in the theorem.

The last essential equation for~$G$ is given by Eq.~\eqref{eq:PushForwardOfT7}.
Collecting coefficients of $\p_x$, $\p_u$ and $\p_g$ in Eq.~\eqref{eq:PushForwardOfT7}, we have that
$\tilde\varphi^{\hat\varphi}(\tilde x)=\varphi_x\hat\varphi$, $\tilde\psi^{\hat\varphi}(\tilde x)=\psi_x\hat\varphi$ and
\begin{equation}\label{eq:ConditionsForArbitraryElementsByAlgabraicMethod2}
\hat\varphi G_x-\hat\varphi_xu_xG_{u_x}+2\hat\varphi_xfG_f+\hat\varphi_{xx}u_xfG_g=
\tilde\varphi^{\hat\varphi}_{\tilde x\tilde x}\tilde u_{\tilde x}F-\tilde\psi^{\hat\varphi}_{\tilde x\tilde x}F,
\end{equation}
respectively.
We proceed in a way analogous to the previous step.
Namely, we express the derivatives $\tilde\varphi^{\hat\varphi}_{\tilde x\tilde x}$
and~$\tilde\psi^{\hat\varphi}_{\tilde x\tilde x}$ via derivatives of~$\hat\varphi$,
substitute the expressions into Eq.~\eqref{eq:ConditionsForArbitraryElementsByAlgabraicMethod2} and
split with respect to derivatives of~$\hat\varphi$ because the function~$\hat\psi$ is arbitrary.
Equating the coefficients of~$\hat\varphi_x$ leads to the equation
\[fG_f=c_1^{-2}\varphi_x^{-1}(c_2u_x+\psi_x\varphi_{xx}-\psi_{xx}\varphi_x).\]

The simultaneous integration of all the equations obtained for~$G$ precisely results in the transformation component for~$g$ from the theorem,
which completes step~4 of the megaideal-based method.

In order to complete the proof, we should realize the last procedure step, which is in fact trivial for the class~\eqref{eq:IbragimovClass}.
We should just check by the direct computation of expressions for transformed derivatives
that any transformation of the form~\eqref{eq:EquivalenceGroupIbragimovClass} maps any equation from the class~\eqref{eq:IbragimovClass}
to an equation from the same class.
\end{proof}

After comparing the equivalence algebra~$\mathfrak g^\sim$ and the equivalence group~$G^\sim$, the following corollary is evident:

\begin{corollary}
A complete list of discrete equivalence transformations of the class~\eqref{eq:IbragimovClass}
that are independent up to combining with each other and with continuous equivalence transformations of this class
is exhausted by the transformations
\begin{align*}
&(t,x,u,u_x,f,g)\mapsto(-t,x,u,u_x,f,g),\\
&(t,x,u,u_x,f,g)\mapsto(t,-x,u,-u_x,f,g),\\
&(t,x,u,u_x,f,g)\mapsto(t,x,-u,-u_x,f,-g).
\end{align*}

\end{corollary}

Theorem~\ref{thm:EquivalenceGroupIbragrimovClass} implies that any transformation~$\mathcal T$ from~$G^\sim$ of the class~\eqref{eq:IbragimovClass} can be represented as the composition
\[
  \mathcal T = \mathscr D^t(c_1)\mathscr P^t(c_0)\mathscr D(\varphi)\mathscr D^u(c_2)\mathscr F^1(c_4)\mathscr F^2(c_3)\mathscr G(\psi),
\]
where the above parameterized equivalence transformations are
\[
\arraycolsep=0ex
\begin{array}{lllllll}
\mathscr P^t(c_0)  \colon\ & \tilde t=t+c_0,\ & \tilde x=x,      \quad& \tilde u=u,       \ & \tilde u_{\tilde x}=u_x,            \ & \tilde f=f,           \ & \tilde g=g,\\
\mathscr D^t(c_1)  \colon\ & \tilde t=c_1t, \ & \tilde x=x,      \quad& \tilde u=u,       \ & \tilde u_{\tilde x}=u_x,            \ & \tilde f=c_1^{-2}f,   \ & \tilde g=c_1^{-2}g,\\
\mathscr D(\varphi)\colon\ & \tilde t=t,    \ & \tilde x=\varphi,\quad& \tilde u=u,       \ & \tilde u_{\tilde x} = u_x/\varphi_x,\ & \tilde f=\varphi^2_xf,\ & \tilde g=g+\varphi_{xx}u_xf/\varphi_x,\\
\mathscr D^u(c_2)  \colon\ & \tilde t=t,    \ & \tilde x=x,      \quad& \tilde u=c_2u,    \ & \tilde u_{\tilde x}=c_2u_x,         \ & \tilde f=f,           \ & \tilde g=c_2g,\\
\mathscr F^1(c_3)  \colon\ & \tilde t=t,    \ & \tilde x=x,      \quad& \tilde u=u+c_3t,  \ & \tilde u_{\tilde x}=u_x,            \ & \tilde f=f,           \ & \tilde g=g,\\
\mathscr F^2(c_4)  \colon\ & \tilde t=t,    \ & \tilde x=x,      \quad& \tilde u=u+c_4t^2,\ & \tilde u_{\tilde x}=u_x,            \ & \tilde f=f,           \ & \tilde g=g+2c_4,\\
\mathscr G(\psi)   \colon\ & \tilde t=t,    \ & \tilde x=x,      \quad& \tilde u=u+\psi,  \ & \tilde u_{\tilde x}=u_x+\psi_x,     \ & \tilde f=f,           \ & \tilde g=g-\psi_{xx}f,
\end{array}
\]
and the nondegeneracy requires that $c_1c_2\varphi_x\ne0$. These transformations are shifts and scalings in $t$, arbitrary transformations in $x$, scalings of $u$, gauging transformations of $u$ with square polynomials in $t$ and arbitrary functions of $x$.

\section{Conclusion}\label{sec:ConclusionAlgebraicMethod}

In this paper we have extended the algebraic method for the computation of complete point symmetry groups of single systems of differential equations to the computation of equivalence groups of classes of such systems. In general, the equivalence group may include both discrete and continuous equivalence transformations. Unlike the infinitesimal method, which aims at finding the corresponding equivalence algebras from which the continuous parts of equivalence groups are computed, the algebraic method allows also constructing discrete equivalence transformations.
We have developed two versions of this method, automorphism-based and megaideal-based, which may be combined depending on the problem under consideration~\cite{card12Ay}.
The effectiveness of both versions rests on the fact that each equivalence transformation induces, via push-forwarding the vector fields that constitute the corresponding equivalence algebra~$\mathfrak g^\sim$, an automorphism of the algebra~$\mathfrak g^\sim$.
This imposes substantial restrictions on the functional form of equivalence transformations
and strongly simplifies, or even makes trivial, the further application of the direct method.

Each of the versions has own advantages and disadvantages. We briefly recall some of them.

The automorphism-based version ensures the maximal use of properties of equivalence algebras within the framework of the algebraic method.
One more benefit of this version is that most continuous equivalence transformations, which are in fact known as easily obtainable from elements of the corresponding equivalence algebra~$\mathfrak g^\sim$, can be factored out in the course of the computation.
At the same time, the main ingredient of the automorphism-based version is finding the entire automorphism group of the algebra~$\mathfrak g^\sim$,
which places a strong limitation on the dimension and/or structure of~$\mathfrak g^\sim$.
This is why usually the automorphism-based version is applied to classes of differential equations whose equivalence algebras are low-dimensional.

The megaideal-based version rests on the invariance of megaideals of equivalence algebras under the associated automorphisms
and hence it does not require the explicit computation of the automorphism groups of equivalence algebras,
which makes it suitable for finding the equivalence groups of classes with infinite-dimensional equivalence algebras.
A disadvantage of the megaideal-based version is
that the equivalence algebra~$\mathfrak g^\sim$ of a class of differential equations~$\mathcal L|_{\mathcal S}$ may admit properties
that induce essential constraints for equivalence transformations of~$\mathcal L|_{\mathcal S}$ but cannot be interpreted in terms of megaideals of~$\mathfrak g^\sim$.

Sometimes disadvantages of either of the versions of the algebraic method can be reduced by combining these versions.

We have restricted ourselves to the computation of the equivalence group of a single physically relevant example, a class of nonlinear wave equations arising in nonlinear elasticity. The equivalence algebra of this class is infinite dimensional, which thus gives a prototypical example for the effectiveness of the megaideal-based method. We plan to write a more extensive paper with further examples on the use of both versions of the algebraic method in the future.

Important further theoretical developments of the present method may include an extension to the computation of extended equivalence groups, which also play a role in the theory of group classification. The study of equivalence groupoids of classes of differential equations using the algebraic method will be another central milestone.

\subsection*{Acknowledgements}

AB is a recipient of an APART Fellowship of the Austrian Academy of Sciences. This research was supported by the Austrian Science Fund (FWF), project P25064 (EDSCB and ROP). The authors are grateful to Vyacheslav Boyko for productive discussions and useful suggestions. ROP is also grateful for the hospitality and financial support provided by the University of Cyprus.


\end{document}